\documentclass{amsart}

\usepackage[left=1.5in, right=1.5in]{geometry}     
\usepackage{graphicx}              
\usepackage{amsmath}               
\usepackage{amsfonts}              
\usepackage{amsthm}                
\usepackage{enumitem}
\usepackage{amssymb}
\usepackage{url}
\usepackage{color}
\usepackage{mathrsfs}
\usepackage{mathtools}
\usepackage{mathabx}
\usepackage{dsfont}

\numberwithin{equation}{section}

\newtheorem{thm}{Theorem}[section]
\newtheorem{lem}[thm]{Lemma}

\theoremstyle{definition}\newtheorem{defn}[thm]{Definition}
\theoremstyle{definition}\newtheorem{example}[thm]{Example}
\theoremstyle{definition}\newtheorem*{remark}{Remark}

\newcommand{\norm}[1]{\left\lVert#1\right\rVert} 
\newcommand{\RR}{\mathbb{R}}

\newcommand{\ZZ}{\mathbb{Z}}
\newcommand{\NN}{\mathbb{N}}
\newcommand{\Intr}{\displaystyle\int}
\newcommand{\Sumn}{\displaystyle\sum}
\newcommand{\Suma}{\Sumn_{\alpha\in R_+}}
\newcommand*\diff{\mathop{}\!\mathrm{d}}

\newcommand{\al}{\alpha}
\newcommand{\alx}{\langle\alpha,x\rangle}

\newcommand{\IntN}{\displaystyle\int_{\RR^N}}
\newcommand{\1}{\mathds{1}}

\begin{document}

\title[Infinite dimensional analysis for Dunkl operators]{Infinite dimensional systems of particles with interactions given by Dunkl operators}
\author{Andrei Velicu}
\address{Andrei Velicu, Department of Mathematics, Imperial College London, Huxley Building, 180 Queen's Gate, London SW7 2AZ, United Kingdom \textit{and} Institut de Math\'ematiques de Toulouse, Universit\'e Paul Sabatier, 118 route de Narbonne, 31062 Toulouse, France}
\email{andrei.velicu@math.univ-toulouse.fr}
\date{}


\begin{abstract}
Firstly we consider a finite dimensional Markov semigroup generated by Dunkl laplacian with drift terms. Using gradient bounds we show that for small coefficients this semigroup has an invariant measure. We then extend this analysis to an infinite dimensional semigroup on $(\RR^N)^{\ZZ^d}$ which we construct using gradient bounds, and finally we study the existence of invariant measures and ergodicity properties.
\end{abstract}

\maketitle

\section{Introduction}

In this paper we study infinite dimensional models of interacting particle systems generated by Dunkl operators. More precisely, we study the Markov semigroup on the infinite dimensional space $(\RR^N)^{\ZZ^d}$ generated by
$$ \mathcal{L} := \sum_{l\in\ZZ^d} (L^{(l)} + L_1^{(l)}).$$
The operators $L^{(l)}$ and $L_1^{(l)}$ are defined by
$$ L^{(l)} = \Delta_k^{(l)} + b^{(l)} \cdot \nabla_k^{(l)},$$
and
$$ L_1^{(l)} = e^{(l)} \cdot \nabla_k^{(l)},$$
where $\Delta_k^{(l)}$ and $\nabla_k^{(l)}$ are the usual Dunkl laplacian and gradient on $\RR^N$ which act only on the $l$ component of $\omega = (\omega_l)_{l\in\ZZ^d} \in (\RR^N)^{\ZZ^d}$. The functions $b^{(l)} :(\RR^N)^{\ZZ^d} \to \RR^N$ depend only on the $\omega_l$ component, while $e^{(l)}:(\RR^N)^{\ZZ^d} \to \RR^N$ depend on the components $\omega_j$ with $|j-l|<R$ for a finite $R$ called the range of interactions. In other words, the operators $L^{(l)}$ are operators that only depend on the $\omega_l$ component and so they define finite dimensional models. In the absence of the interaction terms $e^{(l)}$, these evolve independently of each other. 

Dunkl operators are differential-difference operators defined in terms of a root system. They were introduced in \cite{Dunkl1989} to study special functions associated to root systems, but have since been studied extensively and have found applications outside harmonic analysis. One typical application is in mathematical physics where they were used to study quantum many body problems associated to Calogero-Moser-Sutherland models; see \cite{vDV} for more details. In probability theory, Dunkl operators have been used to define Markov processes with remarkable properties; for an overview of probabilistic aspects of Dunkl theory see \cite{GRY}.

To tackle the problem introduced above, we first study the finite dimensional case of a single operator $L^{(l)}$ acting on $\RR^N$, for which a priori the existence of an invariant measure is unknown. To begin with, we prove gradient bounds for the resulting semigroup. An essential trick here is to consider a symmetrised gradient which takes into account the reflections associated to the root system defining Dunkl operators. This allows us to deal with the resulting reflection terms and to obtain appropriate bounds. Using these bounds we then prove the existence of an invariant measure for the semigroup in the case of small coefficients $k_\alpha$. 

This analysis is then extended to an infinite dimensional setting, by introducing interaction terms. In this case, corresponding to the description above, the construction of the semigroup is not obvious and it requires an approximation procedure. This requires gradient bounds similar to the finite dimensional case, but which are technically more involved. The bounds obtained allow us to prove that the semigroup has an invariant measure and study its ergodicity properties. 

Dissipative dynamics for infinite systems of interacting particles has been studied a long time; see the classic references \cite{Liggett} and \cite{GZ} for more details. For a study of Markov semigroups in infinite dimensional settings see \cite{Z1996} in the case of elliptic operators, \cite{DKZ} for subelliptic operators, \cite{XZ2009,XZ2010} for Levy-type operators and \cite{KOZ} for operators of hypocoercive type. 

Similar analysis for Dunkl-type operators in the infinite dimensional setting but in a simplified context corresponding to the very particular case of a root system $A_1$, was studied in \cite{Z2016}. In this paper we extend these results to the case of general root systems in all dimensions.

The structure of the paper is as follows. In section \ref{SEC:prelims} we introduce Dunkl operators and basic results that will be used later in the paper. In section \ref{SEC:finitedim} we deal with the finite dimensional case, and in section \ref{SEC:invmeasfinite} we present the infinite dimensional case.

\section{Preliminaries} \label{SEC:prelims}

In this section we present a very quick introduction to Dunkl theory. For more details see the survey papers \cite{Rosler} and \cite{Anker}.

We define a root system to be a finite set $R\subset \RR^N\setminus \{0\}$ such that $R \cap \alpha \RR = \{ -\alpha, \alpha\}$ and $\sigma_\alpha(R) = R$ for all $\alpha\in R$. Here $\sigma_\alpha$ is the reflection in the hyperplane orthogonal to the root $\alpha$, i.e.,
$$ \sigma_\alpha x = x - 2 \frac{\alx}{\langle \alpha,\alpha \rangle} \alpha.$$
The group generated by all the reflections $\sigma_\alpha$ for $\alpha\in R$ is a finite group, and we denote it by $G$. 

An invariant function is a map $k:R \to [0,\infty)$, denoted $\alpha \mapsto k_\alpha$, such that $k(\alpha)=k(g\alpha)$ for all $g\in G$ and all $\alpha\in R$. It is possible to write the root system $R$ as a disjoint union $R=R_+\cup (-R_+)$ such that $R_+$ and $-R_+$ lie on different sides of a hyperplane through the origin; in this case, we call $R_+$ a positive subsystem. This decomposition is not unique, but the particular choice of positive subsystem does not make a difference in the definitions below because of the $G$-invariance of the coefficients $k$.

From now on we fix a root system in $\RR^N$ with positive subsystem $R_+$, and a multiplicity function $k$. We also assume without loss of generality that $|\alpha|^2=2$ for all $\alpha \in R$. For $i=1,\ldots, N$ we define the Dunkl operator on $C^1(\RR^N)$ by
$$ T_i f(x) = \partial_i f(x) + \Suma k_\alpha \alpha_i \frac{f(x)-f(\sigma_\alpha x)}{\alx}.$$
An important result, due to Dunkl \cite{Dunkl1989}, is that Dunkl operators commute, i.e.,
$$ T_iT_j = T_j T_i \qquad \forall 1 \leq i,j \leq N.$$

We will denote by $\nabla_k=(T_1,\ldots, T_N)$ the Dunkl gradient, and $\Delta_k = \displaystyle\sum_{i=1}^N T_i^2$ will denote the Dunkl laplacian. Note that for $k=0$ Dunkl operators reduce to partial derivatives, and $\nabla_0=\nabla$ and $\Delta_0=\Delta$ are the usual gradient and laplacian. The Dunkl laplacian can be expressed in terms of the usual gradient and laplacian using the following formula:
\begin{equation} \label{Dunkllaplacian}
\Delta_k f(x) = \Delta f(x) + 2\Suma k_\alpha \left[ \frac{\langle \nabla f(x),\alpha \rangle}{\alx} - \frac{f(x)-f(\sigma_\alpha x)}{\alx^2} \right].
\end{equation}

The following representation formula 
\begin{equation} \label{representationformula}
\frac{f(x)- f(\sigma_\alpha x)}{\alx} = \int_0^1 \alpha \cdot \nabla f(x-t\alx \alpha) \diff t,
\end{equation}
which holds for all roots $\alpha \in R$, shows that Dunkl operators   leave invariant classical spaces of functions, for example $\mathcal{S}(\RR^N)$, the space of Schwarz functions, or $C_c^\infty(\RR^N)$, the space of smooth compactly supported functions.

The natural spaces to study Dunkl operators are the weighted $L^p(\mu_k)$, where $\diff \mu_k=w_k(x) \diff x$ and the weight function is defined as
$$ w_k(x) = \prod_{\alpha\in R_+} |\alx|^{2k_\alpha}.$$
This is a homogeneous function of degree
$$ \gamma := \Suma k_\alpha.$$
We will simply write $\norm{\cdot}_p$ for the norm in $L^p(\mu_k)$. With respect to this weighted measure we have the integration by parts formula
$$ \IntN T_i(f) g \diff\mu_k = - \IntN f T_i(g) \diff\mu_k.$$

One of the main differences between Dunkl operators and classical partial derivatives is that the Leibniz rule does not hold in general. Instead, we have the following. 

\begin{lem} \label{generalisedLeibniz}
If one of the functions $f,g$ is $G$-invariant, then we have the Leibniz rule
$$ T_i(fg) = f T_ig + g T_if.$$
In general, we have
$$ T_i(fg)(x) = T_if(x)g(x) + f(x)T_ig(x) - \Suma k_\alpha \alpha_i \frac{(f(x)-f(\sigma_\alpha x))(g(x)-g(\sigma_\alpha x))}{\alx}.$$
\end{lem}

The theory of Dunkl operators is further enriched by the existence of an intertwining operator, which connects Dunkl operators to partial derivatives, and by the construction of the Dunkl kernel, which acts as a generalisation of the classical exponential function. Using these toolds, it is then possible to define a Dunkl transform, which generalises the classical Fourier transform, with which it shares many important properties. The methods in this paper do not make any use of these notions, so we will not go into further details here; a more complete account can be found in the review papers recommended at the beginning of this section.

\section{The finite dimensional case} \label{SEC:finitedim}

We consider the operator
\begin{equation} \label{findimgenerator} 
L:= \sum_{i=1}^N (T_i^2+b_i T_i)=\Delta_k +b \cdot \nabla_k,
\end{equation}
where each $b_i \in C^1(\RR^N)$ for $i=1,\ldots, N$. We will assume throughout this section that $L$ generates a Markov semigroup on $C_b(\RR^N)$, the space of bounded continuous functions on $\RR^N$.  Let $P_t:=e^{tL}$ be the semigroup generated by $L$ and denote for brevity $f_t=P_tf$.  The carr\'e-du-champ operator associated with the operator $L$ is given by
$$ \Gamma_L(f):= \frac{1}{2} \left( L(f^2) - 2 f L f \right).$$
This should not be confused with the carr\'e-du-champ operator associated with the Dunkl laplacian
$$ \Gamma(f) := \frac{1}{2} (\Delta_k(f^2) - 2f \Delta_k f),$$
which can be computed using the following lemma.

\begin{lem} \label{Dunklcarreduchamp}
We have
\begin{equation*} 
\Gamma(f) = |\nabla f|^2 + \Suma k_\alpha 
 	\left(
 		\frac{f(x)- f(\sigma_\alpha x)}{\alx}
 	\right)^2.
\end{equation*}
\end{lem}

\begin{proof}
Using (\ref{Dunkllaplacian}), we have
\begin{align*}
\Delta_k (f^2) 
& = \Delta (f^2) + 2 \Suma k_\alpha 
	\left( 
		\frac{\langle \nabla(f^2), \alpha \rangle}{\alx} - \frac{f^2(x)-f^2(\sigma_\alpha x)}{\alx^2}	
	\right)
\\
& = 2f \Delta f + 2 |\nabla f|^2 + 4f(x) \Suma k_\alpha 
	\left(
		\frac{\langle\nabla f, \alpha \rangle}{\alx} 
		- \frac{f(x)-f(\sigma_\alpha x)}{\alx^2}
	\right)
\\
&\qquad
 + 2\Suma k_\alpha 
 	\left(
 		\frac{f(x)- f(\sigma_\alpha x)}{\alx}
 	\right)^2
\\
& = 2f\Delta_k f + 2 |\nabla f|^2 + 2\Suma k_\alpha 
 	\left(
 		\frac{f(x)- f(\sigma_\alpha x)}{\alx}
 	\right)^2.
\end{align*}
The expression for $\Gamma(f)$ then follows immediately from this and the definition.
\end{proof}

In order to avoid long expressions, in this paper we use the (non-standard) notation
$$ A_\alpha f(x) := \frac{f(x)-f(\sigma_\alpha x)}{\alx}$$
for the $\alpha$-dependent difference part in the definition of the Dunkl operators. 

\begin{example}
The main example throughout this section is the linear case corresponding to $b(x) = -cx$ for some $c>0$. This case is related to the generalised Ornstein-Uhlenbeck semigroup studied in \cite{RV}, where it was shown that the operator
$$ L_{OU} := \Delta_k - c x \cdot \nabla f$$
generates a Markov semigroup on $C_b(\RR^N)$. 

We note that our operator is related to $L_{OU}$ via
$$ Lf = L_{OU}f - c \Suma k_\alpha (f-f\circ \sigma_\alpha).$$
Since the linear operator $f \mapsto - c \Suma k_\alpha (f-f\circ \sigma_\alpha)$ is clearly bounded on $C_b(\RR^N)$, standard results from general semigroup theory imply that $L$ generates a contraction semigroup $(P_t)_{t \geq 0}$ (the perturbation of a contraction semigroup generator by a bounded operator generates a contraction semigroup, see for example \cite[Theorem 3.1]{Davies1980}). 

It is easy to see that the assumption \eqref{Gammatildecondition} that we impose below on the drifts $b$ will assure the positive maximum principle holds and thus $(P_t)_{t\geq 0}$ is indeed positivity-preserving, so a Markov semigroup (see \cite[Theorem 19.11]{Kallenberg}).

Moreover, adding suitably small perturbations to the linear drifts $b(x)=-cx$ should produce further Markov semigroups with generators of the desired form \eqref{findimgenerator}. 
\end{example}

\begin{remark}
We expect that the class of functions $b$ for which the operator $L$ generates a Markov semigroup is much larger but this is the subject of a different investigation.
\end{remark}

\subsection{Gradient bounds}

The main result of this section is a bound on the symmetrised gradient form
$$ \tilde{\Gamma} (f)(x) := \sum_{g\in G} |\nabla_k f(gx)|^2.$$
This form of the gradient has the advantage that it is $G$-invariant, so we do not have to deal with the reflected terms that will appear in the computations below. 

\begin{thm} \label{findimgradientbounds}
Let $b_i \in C^1(\RR^N)$ such that $\nabla(b_i)$ is bounded for all $i=1,\ldots, N$ and which satisfy the assumptions
\begin{equation} \label{Gammatildecondition}
\frac{2}{\alx^2} - \langle b(x), \alpha \rangle \frac{1}{\alx} \geq 0 \qquad \forall x\in \RR^N, \; \forall \alpha\in R_+
\end{equation} 
and 
\begin{equation} \label{gcondition}
b(gx) = gb(x) \qquad \forall g\in G, \; \forall x\in \RR^N.
\end{equation}

Consider the operator
$$ L = \Delta_k + b(x) \cdot \nabla_k$$
which is assumed to generate a Markov semigroup $(P_t)_{t\geq 0}$ on $C_b(\RR^N)$. 

We then have the bound
\begin{equation} \label{gradientboundsym} 
\tilde{\Gamma}(P_tf) \leq e^{2\eta t} P_t(\tilde{\Gamma}(f)) \qquad \forall t>0,
\end{equation}
where
$$ \eta := \max_{i} \sup_{x\in\RR^N} \partial_i(b_i)(x) 
		+ (N-1) \max_{i\neq j}\norm{\partial_j(b_i)}_\infty 
		+  \sqrt{2} \gamma \max_{\alpha\in R_+} \norm{A_\alpha(b)}_\infty.$$

\end{thm}

\begin{proof}

Firstly, let $g\in G$ and we compute
\begin{align*}
\frac{\diff}{\diff s} P_{t-s}(|\nabla_k f_s|^2 \circ g)
&=\sum_{j=1}^N \frac{\diff}{\diff s} P_{t-s} ((T_jf_s)^2 \circ g)
\\
&=\sum_{j=1}^N P_{t-s} ( -L ((T_jf_s)\circ g)^2) + 2 (T_jf_s)\circ g  \cdot (T_jL f_s) \circ g)
\\
&=P_{t-s} \left( -2 \sum_{j=1}^N \Gamma_L((T_jf_s) \circ g) + I(f_s) \circ g + 2\sum_{j=1}^N (T_jf_s)\circ g \cdot J_g(T_jf_s)\right),
\end{align*}
where
$$ I(h):=2 \nabla_k h \cdot [\nabla_k,L] h
=2 \sum_{i=1}^N \nabla_k h \cdot [\nabla_k,b_iT_i]h $$
and 
$$ J_g(h):= (Lh)\circ g - L(h\circ g).$$
Here and everywhere in the below the notation $[A,B]$ stands for the commutator of the operators $A$ and $B$, i.e., $[A,B]:=AB-BA$.

\vspace{5pt}
\noindent \textbf{Step 1: compute $\Gamma_L(h)$.} We have 
\begin{align*}
\Gamma_L(h) 
&= \frac{1}{2} \left( \Delta_k (h^2) - 2 h \Delta_k h + \sum_{i=1}^N b_i \left[T_i(h^2) - 2h T_ih\right]\right)
\\
&= \Gamma(h) + \frac{1}{2} \sum_{i=1}^N b_i \left[T_i(h^2) - 2h T_ih\right].
\end{align*}
Using Lemmas \ref{Dunklcarreduchamp} and \ref{generalisedLeibniz}, we have
\begin{align*}
&\Gamma_L(h)(x)
\\
&\qquad
= |\nabla h(x)|^2 + \Suma k_\alpha \left( \frac{h(x)-h(\sigma_\alpha x)}{\alx}\right)^2
-\frac{1}{2} \sum_{i=1}^N \Suma k_\alpha \alpha_i b_i(x) \frac{(h(x)-h(\sigma_\alpha x))^2}{\alx}
\\
&\qquad
= |\nabla h(x)|^2 + \frac{1}{2} \Suma k_\alpha (h(x)-h(\sigma_\alpha x))^2 \left(\frac{2}{\alx^2} - \langle \alpha, b(x) \rangle \frac{1}{\alx} \right).
\end{align*}
Thus, assumption (\ref{Gammatildecondition}) assures that $\Gamma_L(h) \geq 0$.

\vspace{5pt}
\noindent \textbf{Step 2: compute $J_g(h)$.} It will be useful here to see $g \in G$ as a matrix, and recall that $gg^T=g^Tg=I$. A simple computation shows that
\begin{equation} \label{gformula1} 
\nabla(h\circ g) = g^T (\nabla h) \circ g.
\end{equation}
Also, note that 
$$ g \sigma_\alpha g^T (x) = x- \langle \alpha, g^T x\rangle g \alpha = \sigma_{g\alpha} (x),$$
so
$$ A_\alpha(h\circ g)(x) 
= \frac{h(gx)-h(g\sigma_\alpha x)}{\alx} 
= \frac{h(gx)-h(\sigma_{g\alpha}(gx))}{\langle g\alpha, gx\rangle} 
= A_{g\alpha} (h) (gx).$$
Thus, since by $G$-invariance we have $k_\alpha=k_{g\alpha}$, we obtain
\begin{equation} \label{gformula2} 
\Suma k_\alpha \alpha A_\alpha(h\circ g) 
= g^T \Suma k_{g\alpha} g\alpha A_{g\alpha} (h) \circ g
= g^T \Suma k_\alpha \alpha A_{\alpha}(h) \circ g,
\end{equation}
where in the last step we simply used a change of variables. Finally, from \eqref{gformula1} and \eqref{gformula2}, we obtain
\begin{equation} \label{gformula}
\nabla_k (h\circ g) = g^T \nabla_k(h)\circ g.
\end{equation}

Using \eqref{gformula}, we can now compute
$$ \Delta_k(h \circ g) = \Delta_k(h)\circ g,$$
and also
$$ J_g(h)(x) = b(gx) \cdot (\nabla_k h)\circ g - b(x) \cdot g^T(\nabla_k h) \circ g = (b(gx) - gb(x)) \cdot (\nabla_k h) \circ g =0,$$
by our assumption \eqref{gcondition}.

\vspace{5pt}
\noindent \textbf{Step 3: estimate $I(h)$.} We first note that, using the commutativity of Dunkl operators, 
$$ [T_j,b_iT_i]
 = [T_j,b_i]T_i + b_i [T_j,T_i]
 = [T_j,b_i]T_i,$$
and
\begin{align*}
[T_j,b_i]T_ih(x) 
& = T_j(b_iT_ih)(x) - b_i(x) T_jT_ih(x)
\\
& = T_j(b_i)(x) T_ih(x) - \Suma k_\alpha \alpha_j \frac{(b_i(x)-b_i(\sigma_\alpha x))(T_ih(x)-T_ih(\sigma_\alpha x))}{\alx}.
\end{align*}
Thus we have
\begin{align*}
I(h)
& = 2 \sum_{i,j=1}^N T_jh [T_j,b_i]T_ih
\\
& = 2 \sum_{i,j=1}^N 
	\left(
		T_j(b_i) T_j h T_ih
		-\Suma k_\alpha \alpha_j A_\alpha(b_i) T_jh(T_i h-(T_i h) \circ \sigma_\alpha)
	\right)
\\
& = 2 \sum_{i,j=1}^N 
	\left(
		\partial_j(b_i) T_j h T_ih
		+ \Suma k_\alpha \alpha_j A_\alpha(b_i) T_jh (T_i h) \circ \sigma_\alpha
	\right)
\\
& = 2 \sum_{i=1}^N \partial_i(b_i) |T_ih|^2 
+ 2 \sum_{i\neq j} \partial_j(b_i) T_ih T_jh 
+ 2 \sum_{i,j=1}^N \Suma k_\alpha \alpha_j A_\alpha(b_i) T_jh (T_i h) \circ \sigma_\alpha.
\end{align*}
Here, in the last step, we simply isolated the terms containing second powers $|T_ih|^2$, which will be necessary in obtaining gradient bounds. To obtain similar second powers from the mixed terms, we use the basic inequality $2xy \leq x^2 +y^2$.  This gives, for the second sum, 
\begin{align*}
2 \sum_{i\neq j} \partial_j(b_i) T_ih T_jh 
&\leq \sum_{i\neq j} \norm{\partial_j(b_i)}_\infty \left[ (T_ih)^2 + (T_jh)^2 \right]
\\
&\leq 2(N-1) \max_{i \neq j} \norm{\partial_j(b_i)}_\infty  |\nabla_kh|^2,
\end{align*}
and for the second sum, recalling that $|\alpha|^2=2$, we obtain
\begin{align*}
2 \sum_{i,j=1}^N \Suma k_\alpha \alpha_j A_\alpha(b_i) T_jh (T_i h) \circ \sigma_\alpha
&=2 \Suma k_\alpha \langle \alpha, \nabla_k h \rangle \langle A_\alpha(b), (\nabla_k h) \circ \sigma_\alpha \rangle 
\\
&\leq \sqrt{2} \max_{\alpha\in R_+} \norm{A_\alpha(b)}_\infty \Suma k_\alpha \left[ |\nabla_k h|^2 + |(\nabla_k h) \circ \sigma_\alpha|^2 \right]
\end{align*}
\vspace{5pt}
\noindent \textbf{Step 4: gradient bounds.} From the previous three steps, summing up the computations over all $g\in G$, we obtain
\begin{multline*}
\frac{\diff}{\diff s} P_{t-s} (\tilde{\Gamma}(f_s)) 
 \leq \sum_{g\in G} P_{t-s} (I(f_s) \circ g)
\\
\leq 
	2\left(\max_{i} \sup_{x\in\RR^N} \partial_i(b_i)(x) 
		+ (N-1) \max_{i\neq j}\norm{\partial_j(b_i)}_\infty 
		+  \sqrt{2} \gamma \max_{\alpha\in R_+} \norm{A_\alpha(b)}_\infty \right)
		P_{t-s} (\tilde{\Gamma}(f_s)).
\end{multline*}

So we have proved that 
$$ \frac{\diff}{\diff s} P_{t-s} (\tilde{\Gamma}(f_s)) \leq 2\eta P_{t-s}(\tilde{\Gamma}(f_s)),$$
where
$$ \eta := \max_{i} \sup_{x\in\RR^N} \partial_i(b_i)(x) 
		+ (N-1) \max_{i\neq j}\norm{\partial_j(b_i)}_\infty 
		+  \sqrt{2} \gamma \max_{\alpha\in R_+} \norm{A_\alpha(b)}_\infty.$$
Integrating this inequality, we obtain
$$ \tilde{\Gamma}(P_tf) \leq e^{2\eta t} P_t(\tilde{\Gamma}(f)),$$
as required.
\end{proof}

\begin{example} \label{exampleBE-x}
If $b(x)=-cx$, for a constant $c>0$, then it is clear that it satisfies the assumption \eqref{Gammatildecondition} of the Theorem. Moreover, $\partial_i(b_j)=-c\delta_{ij}$ and
$$ A_\alpha(b_i)(x) = c\frac{-x_i + (x-\alx \alpha)_i}{\alx} = -c\alpha_i,$$
so
$$ \eta = -c + 0 + \sqrt{2} \gamma \max_{\alpha} |c\alpha| =  -c + 2c\gamma.$$
Thus, if $\gamma <\frac{1}{2}$, the inequality (\ref{gradientboundsym}) proved above is in fact a coercive bound. 
\end{example}

\begin{remark}
It is possible to extend the previous example of linear drifts by adding small perturbations and such that the conditions of the previous Theorem still hold.
\end{remark}

\subsection{Invariant measures} \label{SEC:invmeasfinite}

The strategy for proving existence of invariant measure is to use general results from probability theory such as the Prokhorov and Krylov-Bogoliubov theorems. 
A useful tool in proving the tightness conditions in these theorems is that of a Lyapunov function.

\begin{defn} \label{defnlyapunov}
We say that a smooth function $\rho: \RR^N \to [0,\infty]$ is a \textbf{Lyapunov function} for $L$ if it satisfies the conditions:

\begin{description}

\item[(i)] $\rho^{-1}([0,\infty)) \neq \emptyset$;

\item[(ii)] for any $M>0$, the level set $\{x\in\RR^N : \rho(x) \leq M \}$ is compact;

\item[(iii)] there exist positive constants $C_1, C_2$, with $C_2\neq 0$, such that
\begin{equation} \label{lemmainvmeascondition}
L \rho (x) \leq C_1 - C_2 \rho (x) 
\end{equation}
holds for all $x\in\RR^N$ for which $\rho (x) \neq \infty$.
\end{description}
\end{defn}

We then have the following easy consequence.

\begin{lem} \label{lemmainvmeas}
Let $\rho$ be a Lyapunov function. Then, for any $x\in\RR^N$ for which $\rho(x) \neq \infty$, $P_t\rho(x)$ is bounded in $t$, i.e., there exists a constant $C>0$ such that 
$$ P_t \rho (x) \leq C \qquad \text{ for all } t>0.$$
\end{lem}

Here, $P_t\rho (x)$ should be understood as the limit
$$ P_t\rho(x)=\lim_{a\to\infty} P_t(\rho_a)(x),$$
where $\rho_a (x) = \min \{ \rho(x), a\}$. 

\begin{proof}
Note that by \eqref{lemmainvmeascondition} we have
$$ L\rho_a (y) \leq C_1 - C_2 \rho_a(y)$$
holds for all $y\in\RR^N$ if $a > \frac{C_1}{C_2}$. Since $\rho_a \in C_b(\RR^N)$, then $P_t(\rho_a)$ is well defined and, moreover, we have
$$ \frac{\diff}{\diff t} P_t(\rho_a) = P_t(L\rho_a) \leq C_1 - C_2 P_t(\rho_a).$$
Integrating this inequality, we obtain
$$ P_t(\rho_a)(x) \leq e^{-C_2t} \rho_a(x) + \frac{C_1}{C_2} (1-e^{-C_2t})
\leq \rho_a(x) + \frac{C_1}{C_2}.$$
Taking $a\to \infty$ we have
$$ P_t(\rho) \leq \rho(x) + \frac{C_1}{C_2}.$$
Thus, the conclusion follows by taking $C=\rho(x) + \frac{C_1}{C_2}$ (recall that $x$ is fixed).
\end{proof}

Let us now construct a Lyapunov function as a cutoff of the euclidean distance. To this end, let $\chi: \RR_+ \to \RR_+$ be a smooth function such that 
$$ \chi(t) =
\begin{cases}
0 & \text{if } t \leq 1 
\\
1 & \text{if } t \geq 2.
\end{cases}
$$
Consider then the function
\begin{equation} \label{defnrho} 
\rho(x) = |x| \chi(|x|).
\end{equation}
Then clearly $\rho$ satisfies the conditions (i) and (ii) of Definition \ref{defnlyapunov}. In order to satisfy assumption (\ref{lemmainvmeascondition}) as well, we need to find suitable bounds on $\mathcal{\rho}$. We first compute, using \eqref{Dunkllaplacian},
\begin{align*}
\Delta_k\rho (x) 
& = \Delta \rho(x) + 2\Suma k_\alpha \frac{\langle \nabla\rho (x), \alpha \rangle}{\alx}
\\
& = \sum_{i=1}^N \partial_i\left( \frac{x_i}{|x|} \chi(|x|) + x_i \chi'(|x|) \right)
 + 2\gamma \left(\frac{1}{|x|}\chi(|x|) + \chi'(|x|) \right)
\\
& = \sum_{i=1}^N \left( \frac{1}{|x|}\chi(|x|) - \frac{x_i^2}{|x|^3} \chi(|x|) + \frac{x_i^2}{|x|^2} \chi'(|x|) + \chi'(|x|) + \frac{x_i^2}{|x|} \chi''(|x|) \right)
\\
& \qquad
+ 2\gamma \left(\frac{1}{|x|}\chi(|x|) + \chi'(|x|) \right)
\\
& = (N + 2\gamma -1) \frac{1}{|x|}\chi(|x|)
 + (N + 2\gamma +1) \chi'(|x|) 
 + |x| \chi''(|x|),
\end{align*}
and we notice that, from the choice of $\chi$, there exists a constant $C_1\geq 0$ such that 
$$ \Delta_k\rho \leq C_1.$$
Thus, we have
\begin{align*}
L\rho(x)
& \leq C_1 + \sum_{i=1}^N b_i(x) \partial_i\rho(x)
\\
& = C_1 + \langle x, b(x) \rangle \left( \frac{1}{|x|} \chi(|x|) + \chi'(|x|) \right)
\\
& \leq C_2 + \frac{\langle x, b(x) \rangle}{|x|^2} \rho(x),
\end{align*}
for a constant $C_2\geq 0$. To summarise, we have proved the following Lemma.

\begin{lem}
Assume that there exists a constant $C>0$ such that 
$$ \frac{\langle x, b(x) \rangle}{|x|^2} \leq -C \qquad \text{ for all } x\in\RR^N.$$
Then, the function $\rho$ defined by \eqref{defnrho} is a Lyapunov function.
\end{lem}

\begin{example}
In the case of linear drifts $b(x)=-cx$ with $c>0$, we have
$$ \frac{\langle x, b(x) \rangle}{|x|^2} =-c,$$
so the function $\rho$ does satisfy the assumption of the Lemma, so it is a Lyapunov function. 
\end{example}

We are now ready to prove the main result of this section.

\begin{thm} \label{invmeasthm}
Let $\gamma < \frac{1}{2}$. Assume that there exists $\eta <0$ such that the following inequality holds
$$\tilde{\Gamma}(P_tf) \leq e^{2\eta t} P_t(\tilde{\Gamma}(f)).$$
Moreover, assume that there exists a Lyapunov function $\rho$.

Then, there exists a sequence $(P_{t_l})_{l \geq 0}$ and a probability measure $\nu$ such that for all $f$ bounded and continuous, and for all $x\in\RR^N$ such that $\rho(x) \neq \infty$, we have
$$ P_{t_l} f(x) \to \int f \diff \nu \qquad \text{ as } l \to \infty.$$
Additionally, the measure $\nu$ is invariant for the Markov semigroup $P_t$.  
\end{thm}

\begin{proof}

Fix $x\in \rho^{-1}([0,\infty))$ and define the probability measures
$$ p_t^x (A) := P_t(\1_A)(x).$$

\noindent\textbf{Step 1: $(p_t^x)_{t>0}$ is a tight family.}
By Markov's inequality we have, for any $M >0$, 
\begin{align*}
p_t^x (\{ \rho \geq M \}) 
\leq \frac{1}{M} \int \rho \diff p_t^x
=\frac{1}{M} P_t \rho(x)
\leq \frac{C}{M},
\end{align*}
where in the last step we used Lemma \ref{lemmainvmeas}. This implies that
\begin{equation} \label{invmeasineq}
 p_t^x (\{ \rho \leq M \}) \geq 1-\frac{C}{M}.
\end{equation}
Since the set $\{ \rho \leq M \}$ is compact, and $M>0$ is arbitrary, this shows that the family $(p_t^x)_{t>0}$ is tight. By Prokhorov's theorem, there exists a sequence $(t_l)_{l \geq 0}$ such that $(p_{t_l}^x)_{l \geq 0}$ converges weakly to a probability measure, say $\nu$. 

\vspace{5pt}
\noindent\textbf{Step 2: $p_t^y \Rightarrow \nu$ for all $y\in \rho^{-1}([0,\infty))$.}

Let $\gamma_{x,y}(s)=x+s(y-x)$. We first compute
\begin{align*}
&|P_{t_l}f(y)-P_{t_l}f(x)| 
= \left|\int_0^1 \nabla(P_{t_l}f)(\gamma_{x,y}(s)) \cdot (y-x) \diff s \right|
\\
&\qquad
\leq |x-y| \int_0^1 |\nabla(P_{t_l}f)(\gamma_{x,y}(s))| \diff s
\\
&\qquad
=|x-y| \int_0^1 \left| \nabla_k(P_{t_l}f)(\gamma_{x,y}(s)) - \Suma k_\alpha \alpha \frac{P_{t_l}f(\gamma_{x,y}(s))-P_{t_l}f(\sigma_\alpha(\gamma_{x,y}(s)))}{\langle \alpha, \gamma_{x,y}(s) \rangle} \right| \diff s
\\
&\qquad
\leq |x-y| \int_0^1 |\nabla_k(P_{t_l}f)(\gamma_{x,y}(s))| \diff s 
\\
&\qquad\qquad
	+ \sqrt{2} |x-y| \Suma k_\alpha \int_0^1 \left|\frac{P_{t_l}f(\gamma_{x,y}(s))-P_{t_l}f(\sigma_\alpha(\gamma_{x,y}(s)))}{\langle \alpha, \gamma_{x,y}(s) \rangle}\right| \diff s.
\end{align*}
From our assumption, we have
\begin{equation}
 |\nabla_k(P_t f)| \leq \sqrt{\tilde{\Gamma}(P_tf)} \leq e^{\eta t} \norm{\tilde{\Gamma}(f)}_\infty^{1/2},
\end{equation}
so the above becomes
\begin{equation} \label{invmeasdifference}
\begin{aligned}
|P_{t_l}f(y)-P_{t_l}f(x)| 
&\leq |x-y| e^{\eta t_l} \norm{\tilde{\Gamma}(f)}_\infty^{1/2} 
\\
&\qquad
	+ \sqrt{2} |x-y| \Suma k_\alpha \int_0^1 \left|\frac{P_{t_l}f(\gamma_{x,y}(s))-P_{t_l}f(\sigma_\alpha(\gamma_{x,y}(s)))}{\langle \alpha, \gamma_{x,y}(s) \rangle}\right| \diff s.
\end{aligned}
\end{equation}
Without loss of generality we can assume that $x=0$ and we will use an iteration argument applying (\ref{invmeasdifference}) repeatedly. We first have 
\begin{equation} \label{invmeasdifference0}
\begin{aligned}
|P_{t_l}f(y)-P_{t_l}f(0)|
&\leq |y| e^{\eta t_l} \norm{\tilde{\Gamma}(f)}_\infty^{1/2}
\\
&\qquad 
	+ \sqrt{2} |y| \Suma \frac{k_\alpha}{|\langle \alpha, y\rangle|} \int_0^1 \frac{1}{s} |P_{t_l}f(sy)-P_{t_l}f(\sigma_\alpha(sy))| \diff s.
\end{aligned}
\end{equation}
Applying (\ref{invmeasdifference}) again, noting that $|sy-\sigma_\alpha(sy)| = s \sqrt{2}|\langle \alpha, y\rangle|$, we have
\begin{align*}
&|P_{t_l}f(y)-P_{t_l}f(0)|
\leq |y| e^{\eta t_l} \norm{\tilde{\Gamma}(f)}_\infty^{1/2}  \left( 1+ 2\Suma k_\alpha \right) 
\\
&\qquad 
	+ 2^{3/2} |y| \Suma \sum_{\beta\in R_+} k_\alpha k_\beta \int_0^1 \int_0^1 \left|\frac{P_{t_l}f(\gamma_{\sigma_\alpha(sy),sy}(u))-P_{t_l}f(\sigma_\beta(\gamma_{\sigma_\alpha(sy),sy}(u)))}{\langle \beta, \gamma_{\sigma_\alpha(sy),sy}(u) \rangle} \right| \diff u \diff s.
\end{align*}
The next step is already too difficult to write down, but we can see that we obtain
\begin{align*}
|P_{t_l}f(y)-P_{t_l}f(0)|
&\leq |y| e^{\eta t_l} \norm{\tilde{\Gamma}(f)}_\infty^{1/2}  \sum_{i=0}^\infty \left[ 2^i \sum_{\alpha_1\in R_+} \sum_{\alpha_2\in R_+} \ldots \sum_{\alpha_i\in R_+} k_{\alpha_1} k_{\alpha_2} \cdots k_{\alpha_i}\right]
\\
&=|y| e^{\eta t_l} \norm{\tilde{\Gamma}(f)}_\infty^{1/2}  \sum_{i=0}^\infty (2\gamma)^i.
\end{align*}
Therefore, if $\gamma < \frac{1}{2}$, the infinite sum converges and we have 
\begin{equation} \label{invmeasstep2}
|P_{t_l}f(y) - P_{t_l}f(0)| 
\leq C |y| e^{\eta t_l}  \norm{\tilde{\Gamma}(f)}_\infty^{1/2} .
\end{equation}
Since $\eta<0$, this shows that $p_{t_l}^y \Rightarrow \nu$, as required.

\vspace{5pt}
\noindent\textbf{Step 3: $\nu$ is an invariant measure for $P_t$.}
We will now apply the Krylov-Bogoliubov theorem. We have, using (\ref{invmeasineq}),
\begin{align*}
\mu_T^x(\{ \rho \leq M\}) 
= \frac{1}{T} \int_0^T p_t^x(\{ \rho \leq M\}) \diff t
\geq 1-\frac{C}{M},
\end{align*}
so the family $(\mu_T^x)_{T>0}$ is tight. By the Krylov-Bogoliubov theorem, there exists an invariant measure for $P_t$, say $\mu$. 

Fix a bounded and continuous function $f$ and let $\epsilon >0$. Then, since $p_{t_l}^x \Rightarrow \nu$, there exists $l_1>0$ such that 
$$ |p_{t_l}^x(f) - \nu(f)| < \frac{\epsilon}{2} \qquad \text{ for all } l\geq l_1.$$
On the other hand, since $\mu$ is an invariant measure, we have
\begin{align*}
|p_{t_l}^x(f)-\mu(f)|
\leq \int \left| P_{t_l}f(x) - P_{t_l}f(y) \right| \diff\mu(y) 
\leq C \norm{\tilde{\Gamma}(f)}_\infty^{1/2} e^{\eta t_l} \int |x-y| \diff\mu(y),
\end{align*}
where, in the second step, we used (\ref{invmeasstep2}). Since $\eta <0$, there exists $k_2>0$ such that
$$ |p_{t_l}^x(f) - \mu(f)| < \frac{\epsilon}{2} \qquad \text{ for all } l>l_2.$$
Thus, for all $l>\max \{ l_1, l_2\}$, we have
$$ |\mu(f) - \nu(f)| \leq |p_{t_l}^x(f) - \mu(f)| + |p_{t_l}^x(f) - \nu(f)|<\epsilon,$$
which shows that $\mu=\nu$. This finishes the proof.
\end{proof}

\begin{remark}
We note that we needed to impose conditions for the coefficients $k_\alpha$ to be small both in order to obtain $\eta <0$, as well as in the proof of existence of invariant measure. It would be interesting to find out whether this is just a shortfall of our chosen method, or whether this is a deeper characteristic of Dunkl operators.
\end{remark}

\begin{example} \label{examplelineardriftsmeas}
If $b(x)=-cx$, for some $c>0$, then an invariant measure for the resulting semigroup is given by $\diff\nu(x)=e^{-c|x|^2/2} \diff\mu_k(x)$. Indeed, we have
\begin{align*}
\IntN Lf \diff\nu 
&= \IntN \left(\Delta_k f - c x \cdot \nabla_k f \right) e^{-c|x|^2/2} \diff\mu_k
\\ 
&= - \IntN (\nabla \left(-\frac{c|x|^2}{2} \right) -cx) \cdot \nabla_k f \diff\nu
=0.
\end{align*}
\end{example}

\section{The infinite dimensional case}
 
In this section we work over the infinite dimensional space $\Omega=(\RR^N)^{\ZZ^d}$. An element of this set will be denoted by $\omega=(\omega_l)_{l\in\ZZ^d}$, where each $\omega_l\in\RR^N$. On the lattice $\ZZ^d$ we define the distance $|l-j|=d(l,j):=\Sumn_{i=1}^d |l_i-j_i|$. A cylinder function on $\Omega$ is a smooth function $f:\Omega\to \RR$ that only depends the components $\omega_l$ for $l\in\Lambda$ and $\Lambda\subset \ZZ^d$ is a finite set. The smallest such subset $\Lambda \subset \ZZ^d$ will be denoted $\Lambda(f)$. It is known that the set of cylinder functions on $\Omega$ is dense in $\mathcal{C}(\Omega)$, see for example \cite{GZ}. 

\subsection{Construction of the semigroup}

We want to define an infinite dimensional semigroup on $(\RR^N)^{\ZZ^d}$ with generator 
$$ \mathcal{L}=\Sumn_{l\in\ZZ^d} L^{(l)} + \sum_{l\in\ZZ^d} e^{(l)} \cdot \nabla^{(l)}_k,$$
where 
$$L^{(l)}=\Delta_k^{(l)}+b^{(l)}\cdot \nabla_k^{(l)}$$
is a copy of the operator studied in the previous section, acting only on the $l$ component of $\omega \in (\RR^N)^{\ZZ^d}$. Here $b^{(l)}:\Omega \to \RR^N$ is a function that only depends on the $\omega_l$ component. In addition, $e^{(l)}:\Omega \to \RR^N$ is a $C^1$ function which depends only on $\omega_j$ for $|j-l|<R$ for some fixed $R>0$ (we say that $e^{(l)}$ has finite range of interaction), and such that $e^{(l)}$ and all its first degree derivatives are uniformly bounded. As in the previous section, we assume firstly that for any $l\in\ZZ^d$ we have
\begin{equation}  \label{Gammaconditioninf} 
\frac{2}{\langle \alpha, \omega_l \rangle ^2} - \langle \alpha, b^{(l)}(\omega) \rangle \frac{1}{\langle \alpha, \omega_l \rangle} \geq 0 \qquad \forall \omega \in \Omega, \; \forall \alpha\in R_+
\end{equation}
and
\begin{equation} \label{Gammaconditioninf2} 
\frac{2}{\langle \alpha, \omega_l \rangle ^2} - \langle \alpha, b^{(l)}(\omega) + e^{(l)}(\omega) \rangle \frac{1}{\langle \alpha, \omega_l \rangle} \geq 0 \qquad \forall \omega \in \Omega, \; \forall \alpha\in R_+.
\end{equation}
Secondly, we also assume that for all $l\in\ZZ^d$ we have
\begin{equation} \label{gconditioninf1}
b^{(l)} \circ g^{(l)} = g^{(l)} b^{(l)} \quad \text{ and } \quad e^{(l)} \circ g^{(l)} = g^{(l)} e^{(l)},
\end{equation}
and 
\begin{equation} \label{gconditioninf2}
\quad e^{(\overline{l})} \circ g^{(l)} = e^{(\overline{l})} \quad \forall\; \overline{l} \neq l.
\end{equation}
These two sets of assumptions mirror the conditions \eqref{Gammatildecondition} and \eqref{gcondition}, respectively, of the finite dimensional case. 

The operators $L^{(l)}$ are commutative and, as in the previous section, each generates a finite dimensional semigroup. Thus, their sum $\sum_{l\in\ZZ^d} L^{(l)}$ also generates a semigroup. In this case, the dynamics is given by infinitely many copies of the same model acting independently of each other. Introducing the terms corresponding to the functions $e^{(l)}$ makes the dynamics more interesting as it allows for interactions between the individual diffusions. 

In order to define this semigroup formally, we first consider the truncated operator for some finite $\Lambda \subset \ZZ^d$,
$$ \mathcal{L}_\Lambda := \Sumn_{l\in\ZZ^d} L^{(l)} + \Sumn_{l\in\Lambda} e^{(l)}\cdot \nabla_k^{(l)}.$$
This operator is well defined on the space of cylinder functions on $\Omega$ and we assume that it generates a Markov semigroup $(P_t^{\Lambda})_{t \geq 0}$. Since $\Lambda$ is finite and by the finite range of interactions assumption, then in fact only the diffusions corresponding to the points of $\Lambda^R:= \{l\in\ZZ^d : d(l,\Lambda)<R \}$ interact with each other. Thus, we can split $\mathcal{L}_\Lambda$ into two components
$$ \mathcal{L}_\Lambda = \sum_{l\in \ZZ^d\setminus \Lambda^R} L^{(l)} + \left(\sum_{l\in\Lambda^R} L^{(l)}+ \Sumn_{l\in\Lambda} e^{(l)}\cdot \nabla_k^{(l)}\right).$$
As discussed above, the operators $L^{(l)}$ are commutative, so the first part generates a semigroup of diffusions that do not interact with each other, while the second part is essentially finite dimensional and it can be dealt with as in the previous section. 

Define the carr\'e-du-champ operator
$$ \Gamma_\Lambda(f) := \frac{1}{2} \left(\mathcal{L}_\Lambda(f^2) - 2 f \mathcal{L}_\Lambda f \right).$$
Using computations similar to the finite dimensional case, we have
\begin{align*}
\Gamma_\Lambda(f)(\omega)
&= \sum_{l\in\ZZ^d} |\nabla^{(l)}f(\omega)|^2
+ \frac{1}{2} \sum_{l\in \ZZ^d \setminus \Lambda} k_\alpha A_\alpha^{(l)}(f)^2(\omega) \left[2 - \langle \alpha, b^{(l)}(\omega) \rangle \langle \alpha,\omega \rangle \right]
\\
&\qquad + \frac{1}{2} \sum_{l\in \Lambda} k_\alpha A_\alpha^{(l)}(f)^2(\omega) \left[2 - \langle \alpha, b^{(l)}(\omega)+e^{(l)}(\omega) \rangle \langle \alpha,\omega \rangle \right].
\end{align*}
Here, in keeping with the notation above, $\nabla^{(l)}$ is the usual gradient acting only on the $\omega_l$ component of $\omega \in \Omega$, and $A_\alpha^{(l)}$ is the operator $A_\alpha$ acting on the same component $\omega_l$. From the assumptions \eqref{Gammaconditioninf} and \eqref{Gammaconditioninf2} we have
$$ \Gamma_\Lambda (f) \geq 0.$$

We want to define the infinite semigroup for any cylindrical function $f$ as the limit in the uniform norm
$$ P_tf=\displaystyle\lim_{\Lambda\to\ZZ^d} P_t^{\Lambda}f.$$
In order to prove that this limit exists, it is enough to show that for each increasing sequence $(\Lambda_n)_{n\in\NN}$ of finite subsets of $\ZZ^d$, the sequence $(P_t^{\Lambda_n}f)_{n\in\NN}$ is Cauchy. In other words, fixing a cylindrical function $f$, for finite $\Lambda\subset\Lambda'\subset \ZZ^d$, we want to estimate the norm
$$ \norm{P_t^{\Lambda'}f-P_t^{\Lambda}f}_\infty.$$
We assume also that $\Lambda$ is large enough so that $\Lambda(f)\subset \Lambda$. We construct a sequence $\Lambda_0,\ldots,\Lambda_n$, where $n=|\Lambda'|-|\Lambda|$, such that $\Lambda_0=\Lambda$, $\Lambda_n=\Lambda'$ and for each $j=0,\ldots n-1$, $\Lambda_{j+1}\setminus \Lambda_j = \left\{ h_j\right\}$. Then, we have
\begin{align} \label{cauchyseq1}
\norm{P_t^{\Lambda'}f-P_t^{\Lambda}f}_\infty\leq \Sumn_{j=0}
^{n-1} \norm{P_t^{\Lambda_{j+1}}f-P_t^{\Lambda_j}f}_\infty. 
\end{align}
But we can compute 
$$ P_t^{\Lambda_{j+1}}f-P_t^{\Lambda_j}f 
= \Intr_0^t \frac{\diff}{\diff s} \left[ P_{t-s}^{\Lambda_j}(P_s^{\Lambda_{j+1}}f)\right] \diff s 
= \Intr_0^t P_{t-s}^{\Lambda_j}(\mathcal{L}_{\Lambda_{j+1}}-\mathcal{L}_{\Lambda_j}) P_s^{\Lambda_{j+1}}f \diff s,$$ 
and since $P_t^{\Lambda_j}$ is a Markov semigroup we can estimate
\begin{align*}
|P_t^{\Lambda_{j+1}}f-P_t^{\Lambda_j}f|
\leq \norm{e^{(h_j)}}_\infty \int_0^t P_{t-s}^{\Lambda_j} \left( \left| \nabla_k^{(h_j)}(P_s^{\Lambda_{j+1}}f)\right| \right) \diff s
\end{align*}
Similarly to the previous section, we now introduce the symmetrised gradient form 
\begin{equation} \label{symgraddef}
\tilde{\Gamma}^{(l)}(h)(\omega) 
= \sum_{g\in G} |\nabla_k^{(l)} h(g^{(l)}\omega)|^2,
\end{equation}
defined for all $l\in\ZZ^d$, and where $g^{(l)}$ is the element $g$ of the reflection group $G$ acting on the $\omega_l \in \RR^N$ coordinate of $\omega \in \Omega$. Using this definition, we have
$$ \left| \nabla_k^{(h_j)}(P_s^{\Lambda_{j+1}}f)\right|^2 \leq \tilde{\Gamma}^{(h_j)}(P_s^{\Lambda_{j+1}}f),$$ 
and hence, from the above we deduce
\begin{equation} \label{cauchyseq2}
\norm{P_t^{\Lambda_{j+1}}f-P_t^{\Lambda_j}f}_\infty
\leq \norm{e^{(h_j)}}_\infty \int_0^t \norm{\tilde{\Gamma}^{(h_j)}(P_s^{\Lambda_{j+1}}f)}_\infty^{1/2} \diff s.
\end{equation}
Thus, we need to obtain bounds on the quantity $\norm{\tilde{\Gamma}^{(h_j)}(P_s^{\Lambda_{j+1}}f)}_\infty$. This is done in the following Theorem.

\begin{thm} \label{infsymgradlem}
Let $\Lambda$ be a finite subset of $\ZZ^d$ and consider the Markov semigroup $(P^\Lambda_t)_{t\geq 0}$ generated by the operator
$$\mathcal{L}_\Lambda := \Sumn_{l\in\ZZ^d} L^{(l)} + \Sumn_{l\in\Lambda} e^{(l)}\cdot \nabla_k^{(l)}$$
and denote $f_t:=P_t^\Lambda f$. Assume that 
$$ \eta_l := \max_i \sup_{\omega \in \Omega} \partial_i^{(l)}(b_i^{(l)})(\omega) + (N-1) \max_{m\neq i} \norm{\partial_m^{(l)}(b_i)}_\infty + \sqrt{2} \gamma \max_{\alpha\in R_+} \norm{A_\alpha^{(l)}(b^{(l)})}_\infty <0$$
and
$$ \zeta := \sum_{l\in\ZZ^d} \norm{e^{(l)}}_\infty < \infty.$$
Let $f$ be a cylinder function such that $\Lambda(f)\subset \Lambda$ and let $l\not\in \Lambda(f)$. Let also 
$$ N_{l} =\left[ \frac{d(l,\Lambda(f))}{R}\right] +1 , $$
where the square brackets indicate integer part. Then, for any $\sigma>0$ there exists $\tau\geq 1$ large enough such that if $N_l\geq \tau s$, and we have the bound
\begin{equation} \label{finitespeed}
\norm{\tilde{\Gamma}^{(l)}(f_s)}_\infty 
\leq e^{-2N_l\sigma - 2s\sigma} \sum_{j\in\ZZ^d} \norm{\tilde{\Gamma}^{(j)}(f)}_\infty.
\end{equation}
\end{thm}

\begin{proof}

Fix some $i=1, \ldots, N$. We first look at
\begin{equation} \label{infdimlemstep1}
\begin{aligned}
\partial_u P^\Lambda_{s-u}(|T_i^{(l)}f_u|^2) 
&= P^\Lambda_{s-u}(-\mathcal{L}_\Lambda(|T_i^{(l)}f_u|^2)+2T_i^{(l)} f_u\cdot T_i^{(l)}\mathcal{L}_\Lambda f_u) \\
&= P^\Lambda_{s-u}(-2\Gamma_\Lambda (T_i^{(l)}f_u)+2T_i^{(l)} f_u\cdot [T_i^{(l)},\mathcal{L}_\Lambda]f_u)\\
&\leq P^\Lambda_{s-u} (2T_i^{(l)}f_u\cdot [T_i^{(l)},\mathcal{L}_\Lambda]f_u),
\end{aligned}
\end{equation}
where we used the positivity of $\Gamma_\Lambda$. 

We can check directly that for distinct $l,l'\in \ZZ^d$, the operators $T_i^{(l)}$ and $T_j^{(l')}$ commute for any $i,j=1,\ldots, N$. The same holds for $l=l'$, but here we need to use the fact that Dunkl operators commute. As a consequence (recall also that $b^{(j)}$ depends only on the $\omega_j$ component), we have $[T_i^{(l)},L^{(j)}]=0$ for all $j\neq l$. Thus, we can simplify the commutator
\begin{align*}
 [T_i^{(l)},\mathcal{L}_\Lambda]
&= [T_i^{(l)}, L^{(l)}] + \sum_{j\in\Lambda} [T_i^{(l)}, e^{(j)}\cdot \nabla_k^{(j)}]
\\
&= [T_i^{(l)}, L^{(l)}] + \sum_{j\in\Lambda} [T_i^{(l)}, e^{(j)}]\cdot \nabla_k^{(j)},
\end{align*}
where in the last step we used again the commutativity of $T_i^{(l)}$ with $T_u^{(j)}$. So we next need to compute $[T_i^{(l)}, e^{(j)}]$. We have
\begin{align*}
[T_i^{(l)},e^{(j)}]h
&=\partial_i^{(l)}(he^{(j)}) + \Suma k_\al \al_i A_\al^{(l)}(he^{(j)}) -e^{(j)} \partial_i^{(l)} h - \Suma k_\al \al_i A_\al^{(l)}(h)e^{(j)}  \\
&=h\partial_i^{(l)}(e^{(j)}) - \Suma k_\alpha \al_i [A_\al^{(l)}(h)e^{(j)}-A_\al^{(l)}(he^{(j)})].
\end{align*}
But
\begin{align*}
&A_\al^{(l)}(h)e^{(j)}-A_\al^{(l)}(he^{(j)}) 
\\
&\qquad \qquad \qquad
= \frac{1}{\langle \al, \omega_l\rangle} \left[ e^{(j)}(\omega) (h(\omega)-h(\sigma_\al^{(l)}\omega)) - (h(\omega)e^{(j)}(\omega) - h(\sigma_\al^{(l)}\omega)e^{(j)}(\sigma_\al^{(l)}\omega))\right] \\
&\qquad \qquad \qquad
=-h(\sigma_\al^{(l)}\omega) A_\al^{(l)}(e^{(j)}).
\end{align*}
Thus, we have
$$ [T_i^{(l)},e^{(j)}]h = h\partial_i^{(l)}(e^{(j)}) + \Suma k_\al \al_i h(\sigma_\al^{(l)}\omega) A_\al^{(l)}(e^{(j)}).$$
Therefore, we have obtained
\begin{align} \label{infdimlemstep2}
[T_i^{(l)},\mathcal{L}_\Lambda]h
&=[T_i^{(l)}, L^{(l)}]
+\sum_{j\in\Lambda} \left[\partial_i^{(l)}(e^{(j)})\cdot \nabla_k^{(j)} h + \Suma k_\al \al_i A_\al^{(l)}(e^{(j)})\cdot \nabla_k^{(j)} h(\sigma_\al^{(l)}\omega)\right]. 
\end{align}

Recall the definition of the symmetrised gradient form given in \eqref{symgraddef}. As in the finite case, and due to the assumptions \eqref{gconditioninf1} and \eqref{gconditioninf2},  the computations in \eqref{infdimlemstep1} and \eqref{infdimlemstep2} can be extended to more difficult terms, involving compositions with reflections $g^{(l)}$. For simplicity, we do not include all the details here, but the method is the same as before. Thus, we obtain
\begin{align}
&\frac{\diff}{\diff u} P^\Lambda_{s-u}(\tilde{\Gamma}^{(l)}(f_u)) \notag
\\
&\qquad
\leq \sum_{g\in G} \sum_{i=1}^N P^\Lambda_{s-u} \left(2T_i^{(l)}f_u(g^{(l)}\omega) [T_i^{(l)}, L^{(l)}]f_u(g^{(l)}\omega)\right) \label{infdimterm1}
\\
&\qquad\quad
+ \sum_{g\in G} \sum_{i=1}^N \sum_{m=1}^N \sum_{j \in \Lambda} 2P^\Lambda_{s-u} \left[\left( \partial_i^{(l)}(e^{(j)}_m) T_m^{(j)}f_u T_i^{(l)}f_u \right)\circ g^{(l)} \right] \label{infdimterm2}
\\
&\qquad\quad
+\sum_{g\in G} \sum_{i=1}^N \sum_{m=1}^N \sum_{j \in \Lambda} \Suma 2 k_\alpha \alpha_i P^\Lambda_{s-u}  \left[\left( A_\alpha^{(l)}(e^{(j)}_m)(T_m^{(j)}f_u)\circ \sigma_\alpha^{(l)}T_i^{(l)}f_u\right)\circ g^{(l)}\right]. \label{infdimterm3}
\end{align}
The term in \eqref{infdimterm1} is essentially the finite dimensional case discussed in Theorem \ref{findimgradientbounds} since all operators involved only act on the $l$ component. We thus have
\begin{align*}
\sum_{g\in G} \sum_{i=1}^N P^\Lambda_{s-u} \left(2T_i^{(l)}f_u(g^{(l)}\omega) [T_i^{(l)}, L^{(l)}]f_u(g^{(l)}\omega)\right)
\leq 2 \eta_l P^\Lambda_{s-u}(\tilde{\Gamma}^{(l)}(f_u)),
\end{align*}
where
$$ \eta_l := \max_i \sup_{\omega \in \Omega} \partial_i^{(l)}(b_i^{(l)})(\omega) + (N-1) \max_{m\neq i} \norm{\partial_m^{(l)}(b_i)}_\infty + \sqrt{2} \gamma \max_{\alpha\in R_+} \norm{A_\alpha^{(l)}(b^{(l)})}_\infty.$$
In order to estimate the next two terms, let
$$ E_{l,j} := \max \left\{ \max_{i,m} \norm{\partial_i^{(l)}(e^{(j)}_m)}_\infty, \max_{i,m,\alpha} \norm{A_\alpha^{(l)}(e^{(j)}_m)}_\infty \right\}
\leq \sqrt{2} \max_{i,m} \norm{\partial_i^{(l)}(e^{(j)}_m)}_\infty,$$
where the second inequality follows from \eqref{representationformula}. Clearly $E_{l,j} \geq 0$ for all $l,j\in \ZZ^d$, but note also that $E_{l,j}=0$ if $d(l,j)\geq R$ since $e^{(l)}$ only depends on the components $\omega_j$ with $d(i,j)<R$. The strategy in estimating the remaining terms \eqref{infdimterm2} and \eqref{infdimterm3} is to use the inequality $2xy \leq \epsilon x^2 + \frac{1}{\epsilon} y^2$ to separate the products into sums of squares. We first estimate
\begin{align*}
&\sum_{g\in G} \sum_{i=1}^N \sum_{m=1}^N \sum_{j \in \Lambda} 2 P^\Lambda_{s-u} \left[\left( \partial_i^{(l)}(e^{(j)}_m) T_m^{(j)}f_u T_i^{(l)}f_u \right)\circ g^{(l)} \right]
\\
&\qquad
\leq N \sum_{g\in G} \sum_{j \in \Lambda} E_{l,j} P^\Lambda_{s-u} \left[\epsilon |\nabla_k^{(l)}f_u|^2 \circ g^{(l)} + \frac{1}{\epsilon}|\nabla_k^{(j)}f_u|^2 \circ g^{(l)} \right]
\\
&\qquad
= \epsilon N \left(\sum_{j\in\Lambda} E_{l,j}\right) P^\Lambda_{s-u} (\tilde{\Gamma}^{(l)}(f_u))
+ \frac{1}{\epsilon} N \sum_{g\in G} \sum_{j \in \Lambda} E_{l,j} P^\Lambda_{s-u} \left[|\nabla_k^{(j)}f_u|^2 \circ g^{(l)} \right]
\\
&\qquad
\leq \epsilon N \left(\sum_{j\in\Lambda} E_{l,j}\right) P^\Lambda_{s-u} (\tilde{\Gamma}^{(l)}(f_u))
+ \frac{1}{\epsilon} N \sum_{g\in G} \sum_{j \in \Lambda} E_{l,j} P^\Lambda_{s-u} \left[(\tilde{\Gamma}^{(j)}f_u) \circ g^{(l)} \right],
\end{align*}
and similarly
\begin{align*}
&\sum_{g\in G} \sum_{i=1}^N \sum_{m=1}^N \sum_{j \in \Lambda} \Suma 2 k_\alpha \alpha_i P^\Lambda_{s-u}  \left[\left( A_\alpha^{(l)}(e^{(j)}_m)(T_m^{(j)}f_u)\circ \sigma_\alpha^{(l)}T_i^{(l)}f_u\right)\circ g^{(l)}\right]
\\
&\qquad
\leq \sum_{g\in G} \sum_{j \in \Lambda} \Suma \sqrt{2} N k_\alpha E_{l,j}P^\Lambda_{s-u} \left[\epsilon |\nabla_k^{(l)}f_u|^2 \circ g^{(l)} + \frac{1}{\epsilon} |\nabla_k^{(j)}f_u|^2 \circ (\sigma_\alpha^{(l)} \circ g^{(l)}) \right]
\\
&\qquad
= \sqrt{2}N \gamma \epsilon \left( \sum_{j\in\Lambda} E_{l,j} \right) P^\Lambda_{s-u} (\tilde{\Gamma}^{(l)}(f_u))
+\sqrt{2}N \gamma \frac{1}{\epsilon} \sum_{g\in G} \sum_{j \in \Lambda} E_{l,j} P^\Lambda_{s-u} \left[ |\nabla_k^{(j)}f_u|^2 \circ g^{(l)} \right]
\\
&\qquad
\leq \sqrt{2}N \gamma \epsilon \left( \sum_{j\in\Lambda} E_{l,j} \right) P^\Lambda_{s-u} (\tilde{\Gamma}^{(l)}(f_u))
+ \sqrt{2}N \gamma \frac{1}{\epsilon} \sum_{g\in G} \sum_{j \in \Lambda} E_{l,j} P^\Lambda_{s-u} \left[ \tilde{\Gamma}^{(j)}(f_u) \circ g^{(l)} \right].
\end{align*}

Combining all these results, we have obtained that
\begin{align} \label{infdimineq1}
\frac{\diff}{\diff u} P^\Lambda_{s-u}(\tilde{\Gamma}^{(l)}(f_u))
&\leq 2\tilde{\eta}_l P^\Lambda_{s-u}(\tilde{\Gamma}^{(l)}(f_u))
+  \sum_{g\in G} \sum_{j \in \Lambda} \tilde{E}_{l,j} P^\Lambda_{s-u} \left[ \tilde{\Gamma}^{(j)}(f_u) \circ g^{(l)} \right],
\end{align}
where
\begin{equation} \label{tildeetal} 
\tilde{\eta}_l := \eta_l + \epsilon \frac{N}{2} (1+\sqrt{2}\gamma) \left( \sum_{j\in\Lambda} E_{l,j} \right),
\end{equation}
and
$$ \tilde{E}_{l,j} := \frac{1}{\epsilon}N (1+\sqrt{2}\gamma) E_{l,j}.$$
We choose $\epsilon>0$ such that $\tilde{\eta}_l <0$; this is possible because by our assumption $\eta_l <0$. Multiplying inequality \eqref{infdimineq1} by $e^{-2\tilde{\eta}_l u}$, and integrating with respect to $u$  from $0$ to $s$, we obtain
\begin{align*} 
\tilde{\Gamma}^{(l)}(f_s) 
\leq e^{2\tilde{\eta}_l s} P^\Lambda_s(\tilde{\Gamma}^{(l)}f) 
+ \sum_{g\in G} \sum_{j \in \Lambda} \tilde{E}_{l,j} \int_0^s e^{2\tilde{\eta}_l(s-u)} P^\Lambda_{s-u} \left[ \tilde{\Gamma}^{(j)}(f_u) \circ g^{(l)} \right] \diff u.
\end{align*}
Taking supremum norm over this inequality, and since $\tilde{\eta}_l <0$, we have
\begin{equation} \label{infdimineq2} 
\norm{\tilde{\Gamma}^{(l)}(f_s)}_\infty
\leq \norm{\tilde{\Gamma}^{(l)}f}_\infty 
+ |G| \sum_{j\in\Lambda} \tilde{E}_{l,j} \int_0^s \norm{\tilde{\Gamma}^{(j)}(f_u)}_\infty \diff u.
\end{equation}

Up until this point we have not used the assumption that $l\not\in \Lambda(f)$ so this inequality holds for general $l\in \ZZ^d$. Since $l\not\in \Lambda(f)$, then $f$ does not depend on the $\omega_l$ coordinate and so the first term on the right hand side of inequality \eqref{infdimineq2} vanishes.

We can improve inequality \eqref{infdimineq2} by applying itself iteratively to the term under the integral. At the first step we obtain
\begin{align*}
\norm{\tilde{\Gamma}^{(l)}(f_s)}_\infty 
&\leq |G| \sum_{j\in\Lambda} \tilde{E}_{l,j} \int_0^s \left[ \norm{\tilde{\Gamma}^{(j)}f}_\infty 
+ |G| \sum_{j'\in\Lambda} \tilde{E}_{j,j'} \int_0^u \norm{\tilde{\Gamma}^{(j')}(f_v)}_\infty \diff v \right]\diff u
\\
&\leq s|G| \sum_{j\in\Lambda} \tilde{E}_{l,j} \norm{\tilde{\Gamma}^{(j)}f}_\infty 
+ |G|^2 \sum_{j\in\Lambda}\sum_{j'\in\Lambda} \tilde{E}_{l,j}\tilde{E}_{j,j'} \int_0^s \int_0^u \norm{\tilde{\Gamma}^{(j')}(f_u)}_\infty \diff v \diff u.
\end{align*}
Note that if $d(l,\Lambda(f)) \geq R$ and $d(j,l)<R$, then $j\not\in \Lambda(f)$, so the first sum vanishes again. In fact, we need to apply this method iteratively $N_l = \left[\frac{d(l,\Lambda(f)}{R} \right]+1$ times in order to get a non-trivial contribution. After infinitely many iterations, we obtain 
\begin{align*}
\norm{\tilde{\Gamma}^{(l)}(f_s)}_\infty  
&\leq \left(\frac{(sC_l)^{N_l}}{N_l!} + \frac{(sC_l)^{N_l+1}}{(N_l+1)!} + \ldots \right) \sum_{j\in\ZZ^d} \norm{\tilde{\Gamma}^{(j)}(f)}_\infty  
\\
&\leq \frac{(sC_l)^{N_l}}{(N_l)!} e^{sC_l} \sum_{j\in\ZZ^d} \norm{\tilde{\Gamma}^{(j)}(f)}_\infty ,
\end{align*}
where 
\begin{equation} \label{cldefn}
C_l=|G| \sum_{j\in\Lambda} \tilde{E}_{l,j}.
\end{equation}
Note that the factorials appearing in the denominators above are a result of computing nested integrals. 

Let us estimate the constant appearing in this inequality. Using the fact that for any natural number $n$ we have $n! > \left( \frac{n}{e} \right)^n$ (which follows from the expansion of $e^n$), we obtain
$$ \norm{\tilde{\Gamma}^{(l)}(f_s)}_\infty 
\leq e^{sC_l + N_l(\log sC_l - \log N_l +1)} \sum_{j\in\ZZ^d} \norm{\tilde{\Gamma}^{(j)}(f)}_\infty.$$
Now fix $\sigma >0$. There exists $\tau \geq 0$ large enough such that
\begin{equation} \label{condition}
\log \frac{C_l}{\tau} + \frac{C_l}{\tau} +1 \leq -4 \sigma.
\end{equation}
If, in addition, $N_l \geq \tau s$, then
$$ sC_l + N_l(\log sC_l - \log N_l +1) \leq N_l \left( \frac{C_l}{\tau} + \log \frac{C_l}{\tau} +1 \right) \leq -4 \sigma N_l \leq -2N_l\sigma - 2s\sigma.$$
Finally, this gives
\begin{equation} \label{BEineq}
\norm{\tilde{\Gamma}^{(l)}(f_s)}_\infty 
\leq e^{-2N_l\sigma - 2s\sigma} \sum_{j\in\ZZ^d} \norm{\tilde{\Gamma}^{(j)}(f)}_\infty,
\end{equation}
as required.
\end{proof}

Now that we have proved the required bounds on the symmetrised gradient, we can go back to the definition of the infinite dimensional semigroup. Fix $\sigma>0$ and let $\tilde{N}=\left[ \frac{d(\ZZ^d\setminus\Lambda, \Lambda(f))}{R}\right]+1$. By definition, we have $N_{h_j}\geq \tilde{N}$ for all $j=0,\ldots, n-1$. From (\ref{cauchyseq1}) and (\ref{cauchyseq2}), with the help of Theorem \ref{infsymgradlem}, we obtain
\begin{equation} \label{PtLambdadifference}
\begin{aligned}
\norm{P_t^{\Lambda'}f-P_t^{\Lambda}f}_\infty
&\leq \left( \Sumn_{l\in\ZZ^d} \norm{\tilde{\Gamma}^{(l)}(f)}_\infty \right)^{1/2} \Sumn_{j=0}^{n-1} \norm{e^{(h_j)}}_\infty \Intr_0^t e^{-\sigma N_{h_j}-\sigma s} \diff s \\
&\leq \zeta e^{-\sigma \tilde{N}} \frac{1-e^{\sigma t}}{\sigma}  \left( \Sumn_{l\in\ZZ^d} \norm{\tilde{\Gamma}^{(l)}(f)}_\infty \right)^{1/2},
\end{aligned}
\end{equation}
which holds for all $\tilde{N} \geq \tau t$, where $\tau \geq 1$ is fixed. Recall that $\zeta = \sum_{l\in\ZZ^d} \norm{e^{(l)}}_\infty$ was assumed to be finite.

This concludes the proof that for each increasing sequence $(\Lambda_n)_{n\geq 1}$ of finite subsets of $\ZZ^d$, and for each cylindrical function $f$, the sequence $(P_t^{\Lambda_n}f)_n$ is Cauchy. Therefore, we have proved the following Theorem.

\begin{thm} \label{defninfsemigroup}
Consider the operator
$$ \mathcal{L} := \sum_{l\in\ZZ^d} (\Delta_k^{(l)} + b^{(l)} \cdot \nabla_k^{(l)} )
+ \sum_{l\in\ZZ^d} e^{(l)} \cdot \nabla_k^{(l)}.$$
Assume that the coefficients $b^{(l)}$ and $e^{(l)}$ are uniformly bounded such that
$$ \eta_l := \max_i \sup_{\omega \in \Omega} \partial_i^{(l)}(b_i^{(l)})(\omega) + (N-1) \max_{m\neq i} \norm{\partial_m^{(l)}(b_i)}_\infty + \sqrt{2} \gamma \max_{\alpha\in R_+} \norm{A_\alpha^{(l)}(b^{(l)})}_\infty <0$$
and
$$ \zeta := \sum_{l\in\ZZ^d} \norm{e^{(l)}}_\infty < \infty.$$
Then, for any cylinder function $f$, the following limit exists in the uniform norm
$$ P_tf:=\displaystyle\lim_{\Lambda\to\ZZ^d} P_t^\Lambda f,$$
and it defines a Markov semigroup $(P_t)_{t \geq 0}$ on $C_b(\Omega)$, with infinitesimal generator given by $\mathcal{L}$.
\end{thm}

\subsection{Invariant measure}

In this section we discuss the existence of an invariant measure for the infinite dimensional semigroup that we have just defined. The strategy is the same as in section \ref{SEC:invmeasfinite}, and we start by defining a Lyapunov function $\rho$ that is suitable to our new setting.

Let $(a_l)_{l\in\ZZ^d} \subset (0,\infty)$ be an absolutely convergent series, i.e.,
$$ \sum_{l\in\ZZ^d} a_l < \infty.$$
We assume that for each $l\in\ZZ^d$ there exists a function $\rho_l : \Omega \to (0,\infty)$, which depends only on the $\omega_l$ coordinate, and such that $\omega_l \mapsto \rho_l(\omega)$ is a Lyapunov function for the operator $L^{(l)}$ on $\RR^N$. Moreover, assume that there exist constants $C_1 \geq 0$ and $C_2 >0$ such that 
\begin{equation} \label{assumptionrhol} 
L^{(l)} \rho_l + e^{(l)}\cdot \nabla_k^{(l)} \rho_l \leq C_1 - C_2 \rho_l \qquad \forall l\in\ZZ^d.
\end{equation}
For each $0<r\leq \infty$, define the set
$$ \Omega_r :=\left\{ \omega \in \Omega : \sum_{l\in\ZZ^d} a_l \rho_l(\omega) < r \right\},$$
and consider the function $\rho: \Omega_\infty \to (0,\infty)$ given by
$$ \rho(\omega) := \sum_{l\in\ZZ^d} a_l \rho_l(\omega).$$
Using condition \eqref{assumptionrhol} and the fact that $\rho_l$ only depends on the $\omega_l$ coordinate, we have
\begin{equation} \label{infdimlyapunov}
\mathcal{L} \rho 
= \sum_{l\in\ZZ^d} a_l (L^{(l)} \rho_l + e^{(l)}\cdot \nabla_k^{(l)} \rho_l)
\leq \sum_{l\in\ZZ^d} a_l (C_1 - C_2 \rho_l) 
= C_1 \sum_{l\in\ZZ^d} a_l - C_2 \rho.
\end{equation}
Then $\rho$ is a Lyapunov function for $\mathcal{L}$.

\begin{remark}
We could take, as in section \ref{SEC:invmeasfinite}, $\rho_l(\omega) = |\omega_l| \chi(|\omega_l|)$ for a cut-off function $\chi$. If $b^{(l)}$ satisfies
$$ \frac{\langle \omega_l, b^{(l)}(\omega)\rangle}{|\omega_l|^2} \leq - C \qquad \forall \omega_l\in\RR^N,$$
and if the interaction coefficients $e^{(l)}$ are sufficiently small, then more generally $\rho_l$ satisfies the assumption \eqref{assumptionrhol}.
\end{remark}

\begin{remark}
We can relax assumption \eqref{assumptionrhol} to 
$$ L^{(l)} \rho_l \leq C_1 - C_2 \rho_l,$$
and
$$ e^{(l)} \cdot \nabla_l^{(l)} \rho_l \leq C_3 - \sum_{j\in\ZZ^d} \epsilon_{l,j} \rho_j,$$
for suitably small positive numbers $\epsilon_{l,j}$. With a little more effort, the proof below can be changed to accommodate this case.
\end{remark}

We are ready to prove the main result of this section.

\begin{thm}
Let $\rho$ be defined as above, and let $(P_t)_{t\geq 0}$ be the Markov semigroup with generator $\mathcal{L}$ defined in Theorem \ref{defninfsemigroup}. Then, for any $\omega\in \Omega_\infty$ there exists a subsequence $(P_{t_j})_{j\geq 0}$, and a probability measure $\nu_\omega$ such that for all $f\in C(\Omega)$ bounded cylinder functions, we have 
$$ P_{t_j} f(\omega) \to \int f \diff\nu_\omega \qquad \text{ as } j\to\infty.$$
Furthermore, $\nu_\omega(\Omega_\infty)=1$.
\end{thm}

\begin{proof}
Fix $\omega\in \Omega_\infty$ and define
$$ p_t^\omega(A) := P_t(\1_A)(\omega).$$
Using the same approach as in Lemma \ref{lemmainvmeas}, from \eqref{infdimlyapunov} we can show that there exists a $C>0$ such that 
\begin{equation} \label{infdimpt} 
P_t \rho (\omega) \leq C \qquad \text{ for all } t>0.
\end{equation}
As in the proof of Theorem \ref{invmeasthm}, using \eqref{infdimpt} and Prokhorov's theorem, we deduce that there exists a subsequence $(p_{t_j}^\omega)_{j\geq 0}$ which converges weakly to some probability measure $\nu_\omega$. 

Moreover, using Markov's inequality we have for any $r>0$ that
$$ 0\leq 1 -\nu_\omega(\Omega_r) \leq \frac{1}{r} \int \rho \diff\nu_\omega 
=\frac{1}{r} \sup_{l\in\NN} P_{t_l}\rho (\omega) \leq \frac{C}{r},$$
where in the last step we used \eqref{infdimpt}. Taking $r\to\infty$, we obtain $\nu_\omega(\Omega_\infty)=1$. 
\end{proof}

For any $\epsilon >0$ for which we have $\sum_{l\in\ZZ^d} (1+|l|)^{-\epsilon} <\infty$, define also the set 
$$\tilde{\Omega}_\epsilon := \left\{ \omega\in \Omega : \sum_{l\in\ZZ^d} \frac{|\omega_l|}{(1+|l|)^\epsilon} < \infty \right\}.$$
Note that for $\rho_l$ defined as in the previous remark, and with $a_l = (1+|l|)^{-\epsilon}$, then $\tilde{\Omega}_\epsilon = \Omega_\infty$. 

\begin{remark}
From the proof below it will be clear that instead of the weights $(1+|l|)^{-\epsilon}$ we could take more general weights in the definition of $\tilde{\Omega}_\epsilon$, even of exponential decay.
\end{remark}

\begin{thm}
Assume that $\gamma <\frac{1}{2}$. Let $(P_t)_{t\geq 0}$ be the Markov semigroup defined in Theorem \ref{defninfsemigroup} and assume in addition that
$$ \tilde{\eta} := \sup_{l\in\ZZ^d} \tilde{\eta}_l <0 ,$$
and
$$ \tilde{C} := \sup_{l\in\ZZ^d} C_l \leq -2\tilde{\eta},$$
where the constants $\tilde{\eta_l}$ and $C_l$ were defined in \eqref{tildeetal} and \eqref{cldefn}, respectively.

For any bounded cylinder function $f$ and any $\omega, \omega'\in \tilde{\Omega}_\epsilon$, there exists $t_0>0$ and a constant $C(f,\omega,\omega')<\infty$, such that
$$ |P_tf(\omega) - P_tf(\omega')| \leq C(f,\omega, \omega') e^{-ct},$$
where $c>0$ is independent of $f, t, \omega, \omega'$.
\end{thm}

\begin{proof}
Fix a cylinder function $f$ and $\omega,\omega' \in \tilde{\Omega}_\epsilon$. Fix $\Lambda \subset \ZZ^d$ finite and such that $\Lambda(f) \subset \Lambda$. We then have
\begin{equation} \label{3ineqinfdimmeas}
\begin{aligned}
&|P_{t} f(\omega) - P_{t_j}f(\omega')|
\\
&\qquad \leq |P_{t} f(\omega) - P_{t}^\Lambda f(\omega)|
+|P_{t}^\Lambda f(\omega) - P_{t}^\Lambda f(\omega')|
+|P_{t}^\Lambda f(\omega') - P_{t}f(\omega')|.
\end{aligned}
\end{equation}
Using the computations in \eqref{PtLambdadifference}, there exists $r>0$ such that if $\Lambda= B_{rt}(0) = \{l\in\ZZ^d : |l| < rt\}$,  and letting $\Lambda'\to\infty$, then we have
\begin{equation} \label{infdimdecay1}
\norm{P_t f - P_t^\Lambda f}_\infty \leq C_1(f) e^{-c_1 t},
\end{equation}
which holds for all $t \geq t_0$, where $t_0$ is chosen such that $\Lambda(f) \subset \Lambda$.
This deals with the first and the third terms on the right hand side of the inequality \eqref{3ineqinfdimmeas} above. We are thus only left to study the middle term.

If $\omega$ and $\omega'$ differ by just one coordinate, say $\omega_j=\omega'_j$ for all $j\neq l$, then let $\gamma_{\omega_l,\omega'_l}(s) = \omega_l + s(\omega'_l-\omega_l)$. We can then compute
\begin{align*}
|P_{t}^\Lambda f(\omega) - P_{t}^\Lambda f(\omega')| 
&= \left| \int_0^1 \nabla^{(l)}(P_{t}^\Lambda f)(\gamma_{\omega_l,\omega'_l}(s)) \cdot(\omega'_l-\omega_l) \diff s\right|
\\
&\leq |\omega'_l-\omega_l| \int_0^1 |\nabla^{(l)}(P_{t}^\Lambda f)(\gamma_{\omega_l,\omega'_l}(s))| \diff s
\\
&\leq |\omega'_l-\omega_l| \norm{\tilde{\Gamma}^{(l)}(P_{t}^\Lambda f)}_\infty^{1/2}
\\
&\qquad
+ \sqrt{2} |\omega'_l-\omega_l| \Suma k_\alpha \int_0^1 \left|\frac{P_{t}^\Lambda f(\gamma_{\omega_l,\omega'_l}(s))-P_{t}^\Lambda f(\sigma_\alpha(\gamma_{\omega_l,\omega'_l}(s)))}{\langle \alpha, \gamma_{\omega_l,\omega'_l}(s))\rangle} \right|.
\end{align*}
Using the same method as in Step 2 in the proof of Theorem \ref{invmeasthm}, since $\gamma <\frac{1}{2}$, then we obtain
\begin{align} \label{ptdiffstep1}
|P_{t}^\Lambda f(\omega) - P_{t}^\Lambda f(\omega')|
\leq C_2 |\omega_l-\omega'_l| \norm{\tilde{\Gamma}^{(l)}(P_{t}^\Lambda f)}_\infty^{1/2}.
\end{align}

Going back to the proof of \eqref{infdimineq2}, note that we have in fact
\begin{equation*} 
\norm{\tilde{\Gamma}^{(l)}(f_t)}_\infty
\leq e^{2\tilde{\eta} t} \norm{\tilde{\Gamma}^{(l)}f}_\infty 
+ |G| \sum_{j\in\Lambda} \tilde{E}_{l,j} \int_0^t \norm{\tilde{\Gamma}^{(j)}(f_u)}_\infty \diff u.
\end{equation*}
Using the same iteration procedure as above, this implies that
\begin{equation} \label{infdimgradboundstep1}
\norm{\tilde{\Gamma}^{(l)}(f_t)}_\infty
\leq e^{2\tilde{\eta} t + tC_l} \sum_{j\in\ZZ^d} \norm{\tilde{\Gamma}^{(j)}(f)}_\infty.
\end{equation}
By our assumption $2\tilde{\eta}+C_l \leq 2 \tilde{\eta} + \tilde{C} < 0$, so this inequality provides the sufficient decay required in \eqref{ptdiffstep1}. 

We now turn to the case of general $\omega, \omega'\in \tilde{\Omega}_\epsilon$. First, we note that if $\Lambda(f)\subset \Lambda$, then by \eqref{infdimineq2} we have $\Lambda(P_t^\Lambda f) \subset \Lambda^R := \{ l \in \ZZ^d : d(l,\Lambda)<R\}$, where $R$ is the range of interaction. But $\Lambda^R$ is a finite set, say $|\Lambda^R|=n$, and let $l_1, \ldots, l_n$ be its distinct elements. We construct the sequence $\omega^0= \omega$, $\omega^{i+1}_j =\omega^i_j$ for all $j\neq l_i$, and $\omega^{i+1}_{l_i} = \omega'_{l_i}$. In other words, each two consecutive terms of the sequence differ in just one coordinate, $l_i$, and $\omega^n_j = \omega'_j$ for all $j\in \Lambda^R$. Thus, by the observation above, $P_t^\Lambda(\omega^n)=P_t^\Lambda(\omega')$, and so we have
$$ |P_{t}^\Lambda f(\omega)- P_t^\Lambda f(\omega')| 
\leq \sum_{i=0}^{n-1} |P_{t}^\Lambda f(\omega^{i+1})- P_t^\Lambda f(\omega^i)|.$$
We estimate each term in this sum using the same method as above. More precisely, using \eqref{ptdiffstep1} and \eqref{infdimgradboundstep1}, we have
\begin{align}
\sum_{i=0}^{n-1} |P_{t}^\Lambda f(\omega^{i+1})- P_t^\Lambda f(\omega^i)|
&\leq C_2 \sum_{i=0}^{n-1} |\omega_{l_i} - \omega'_{l_i}| e^{(\tilde{\eta}+\tilde{C}/2)t}  \left(\sum_{j\in\ZZ^d} \norm{\tilde{\Gamma}^{(j)}(f)}_\infty\right)^{1/2} \notag
\\
&\leq C_3(f) e^{(\tilde{\eta}+\tilde{C}/2)t}  \sum_{l\in B_{rt+R}(0)} \frac{|\omega_l|+|\omega'_l|}{(1+|l|)^\epsilon} (1+|l|)^\epsilon \notag
\\
&\leq C_3(f)(C_\omega + C_{\omega'} ) e^{(\tilde{\eta}+\tilde{C}/2)t}  (1 + |R+rt|)^\epsilon \label{infdimdecay2}
\end{align}
where $C_\omega = \sum_{l\in\ZZ^d} \frac{|\omega_l|}{(1+|l|)^\epsilon}$ and $C_{\omega'}= \sum_{l\in\ZZ^d} \frac{|\omega'_l|}{(1+|l|)^\epsilon}$ which are both finite since $\omega,\omega'\in \tilde{\Omega}_\epsilon$.

From \eqref{infdimdecay1} and \eqref{infdimdecay2} we obtain the decay we require.
\end{proof}

\noindent \textbf{Acknowledgements.} The author wishes to thank Boguslaw Zegarlinski for introducing him to the problem and for useful advice. Financial support from EPSRC is also gratefully acknowledged. 


\bibliographystyle{plain}
\bibliography{ref}

\end{document}